\documentclass[submission,copyright,creativecommons]{eptcs}
\providecommand{\event}{FROM 2023} 
\usepackage{graphicx}
\usepackage{cite}
\usepackage{amsmath,amssymb,amsfonts,bm}
\usepackage{algorithmic}
\usepackage{multirow}
\usepackage{multicol}
\usepackage{textcomp}
\usepackage{xcolor}
\def\BibTeX{{\rm B\kern-.05em{\sc i\kern-.025em b}\kern-.08em
    T\kern-.1667em\lower.7ex\hbox{E}\kern-.125emX}}
    
\usepackage{graphicx}
\usepackage{listings}
\usepackage{amsmath}
\usepackage{mathtools}
\usepackage{amssymb}
\usepackage{tikz}
\usepackage{tikz-cd}
\usetikzlibrary{positioning}
\usetikzlibrary{calc}
\tikzset{stretch/.initial=1}
\newcommand\drawloop[4][]%
   {\draw[shorten <=0pt, shorten >=0pt,#1]
      ($(#2)!\pgfkeysvalueof{/tikz/stretch}!(#2.#3)$)
      let \p1=($(#2.center)!\pgfkeysvalueof{/tikz/stretch}!(#2.north)-(#2)$),
          \n1= {veclen(\x1,\y1)*sin(0.5*(#4-#3))/sin(0.5*(180-#4+#3))}
      in arc [start angle={#3-90}, end angle={#4+90}, radius=\n1]%
   }

\usepackage{amsthm}
\usepackage{url}

\usepackage{tikz}

\theoremstyle{definition}
\newtheorem{definition}{Definition}[section]

\newtheorem{theorem}{Theorem}

\usepackage{iftex}

\ifpdf
  \usepackage{underscore}         
  \usepackage[T1]{fontenc}        
\else
  \usepackage{breakurl}           
\fi

\title{A Parallel Dynamic Epistemic Perspective over\\ Muddy Children Puzzle}
\author{Bogdan Macovei
\institute{Department of Computer Science \\ Faculty of Mathematics and Computer Science\\
University of Bucharest, 
Academiei 14, 010014 Bucharest, Romania
\\ \\   
Research Center for Logic, Optimization and Security (LOS) \\ Department of Computer Science, Faculty of Mathematics and Computer Science, \\ University of Bucharest, Academiei 14, 010014 Bucharest, Romania\\
\url{los.cs.unibuc.ro}}
\email{bogdan.macovei@fmi.unibuc.ro}
}
\def\titlerunning{A Parallel Dynamic Epistemic Logic}
\def\authorrunning{B. Macovei}
\begin{document}
\maketitle
\begin{abstract}
Epistemic protocols represents a current field of interest, with numerous approaches still being studied.
In this paper we formalize parallel sessions of the The Muddy Children Puzzle using Public Observation Logic, a system that allows epistemic update. We consider agents with roles in multiple sessions and the information
update in all parallel sessions as new information is discovered in any particular session.
\end{abstract}

\section{Introduction}

In this article, we introduce a formalism to analyze the parallel execution of actions in the epistemic protocol of the \textit{Muddy Children Puzzle}, previously modeled using \textit{public announcement logic} in \cite{vanDitmarsch}. We consider agents with roles in multiple sessions and the information update in all parallel sessions as new information is discovered in any particular session. Although the formalism is specifically tailored to model this protocol, it can be generalized since its theoretical framework is based on \cite{hidden}, where a logical system starting from generic actions is defined. Our goal is to investigate how parallel sessions interact and how agents' information is modified not only within a single session but also in a parallel setting.

In the second section, we introduce the \textit{public observation logic} system, based on \cite{hidden}; in the third section, we model the \textit{Muddy Children Puzzle} protocol using this formalism, starting from the description in \cite{vanDitmarsch}; in the fourth section, we construct the parallel sessions and define a method of information propagation between these sessions to maintain consistency in the main parallel model; in the fifth section, we provide an example of such a parallel framework where we study the speed of information propagation and how the number of actions required to reach a result decreases with this update. Finally, we present some conclusions and future directions for further research.

\section{Preliminaries}

In this preliminary section, we will present the main results from \cite{hidden}, where a new logical system, \textit{public observation logic}, is defined. This system is built upon notions such as \textit{observations} and \textit{expectation models}, which are based on \textit{epistemic models}. These concepts will be used in the formalism we will propose later to define a modeling for the \textit{Muddy Children Puzzle}.

Let $\text{\footnotesize AGENT}$ be a finite set of agents, $\Phi$ be a set of formulas, with $\Phi_0 \subseteq \Phi$ the set of atomic propositions.

\begin{definition}[Epistemic model]
 An \textit{epistemic model} $\mathcal{M}$ is $\mathcal{M} = (W, R, V)$, where $W$ is a non-empty set of worlds, $R \subseteq W^2$ is a binary accessibility relation, $R = \{ R_a \mid a \in \text{\footnotesize AGENT}\}$, and $V : \Phi \to \mathcal{P}(W)$ is the valuation function, assigning to each formula the set of worlds where formula is true.
\end{definition}

Additionally, the authors of \cite{hidden} introduce \textit{observation expressions}, as regular expressions over a finite set of actions, named $\Sigma$. 

\begin{definition}[Observation expressions]
    Given a finite set of action symbols $\Sigma$, the language $\mathcal{L}_{obs}$ of \textit{observation expressions} is defined by: 
    \begin{align}
        \pi ::= \delta \mid
        \varepsilon \mid
        \pi \cdot \pi \mid 
        \pi + \pi \mid 
        \pi^* 
    \end{align}
    where $\delta$ is the empty action, with a corresponding empty set $\emptyset$ of observations, the constant $\varepsilon$ represents the empty string, and $a \in \Sigma$.
\end{definition}

\begin{definition}[Observations]
    Given an observation expression $\pi$, the corresponding \textit{set of observations}, that is denoted by $\mathcal{L}(\pi)$ is the set of finite string over $\Sigma$, defined as follows: 
    
    \begin{tabular}{ll}
    $\mathcal{L}(\delta) = \emptyset$ & $\mathcal{L}(\pi_1 \cdot \pi_2) = \{ vw \mid v \in \mathcal{L}(\pi_1) \text{ and } w \in \mathcal{L}(\pi_2) \}$ \\ 
    $\mathcal{L}(\varepsilon) = \{\epsilon\}$ & $\mathcal{L}(\pi_1 \cup \pi_2) = \mathcal{L}(\pi_1) + \mathcal{L}(\pi_2)$ \\ 
    $\mathcal{L}(a) = \{a\}$ & $\mathcal{L}(\pi^*) = \{ \varepsilon \} \cup \bigcup_{n > 0} \mathcal{L}(\underbrace{\pi ; \pi ; \dots ; \pi}_{n \text{ times}})$
\end{tabular}
\end{definition}

\begin{definition}[Epistemic expectation model]
    An \textit{epistemic expectation model} $\mathcal{M}_{exp}$ is a quadruple $\mathcal{M}_{exp} = (W, R, V, Exp)$, where $(W, R, V)$ is an epistemic model, also named the \textit{epistemic skeleton} of $\mathcal{M}_{exp}$, and $Exp : W \to \mathcal{L}_{obs}$ is an expected observation function, assigning to each state an observation expression $\pi$, such that $\mathcal{L}(\pi) \neq \emptyset$. 
\end{definition}

With these notions defined, there is introduced a new formalism, named \textit{Public observation logic}, that is a dynamic logic with knowledge operators, made to reason about knowledge via the matching of observations and expectations. 

\begin{definition}[Update by observation]
    Let $\alpha$ be an observation over $\Sigma$, and let $\mathcal{M} = (W, R, V, Exp)$ be an epistemic expectation model. The updated model $\mathcal{M}|_\alpha$ is defined as $\mathcal{M}|_\alpha = (W', R', V', Exp')$, where $W' := \{ w \mid \mathcal{L}(Exp(w) - \alpha) \neq \emptyset\}$, $R' = R|_{W'^2}$, $V' = V|_{W'}$, and $Exp'(w) = Exp(w) - \alpha$. Here, $\pi - \alpha$ is defined as $\pi - \alpha = \{ \beta \mid \alpha \beta \in \mathcal{L}(\pi) \}$, the right residuation with respect to the monoid $(\Sigma^*, \cdot, \varepsilon)$.
\end{definition}

\begin{definition}[Public observation logic]
    The formula $\varphi$ of POL (Public Observation Logic) are given by the following BNF:
    \begin{align}
    \varphi ::= \top \mid 
        p \mid
        \neg \varphi \mid
        \varphi \land \varphi \mid 
        K_a \varphi \mid 
        [\pi] \varphi
\end{align}
where $p \in \Phi_0$ and for any $a \in \text{\footnotesize AGENT}$, $K_a$ is the knowledge operator corresponding to agent $a$.
\end{definition}

\begin{definition}[Truth definition for POL]
    Given an epistemic expectation model $\mathcal{M} = (W, R, V, Exp)$, a state $w \in W$ and a POL-formula $\varphi$, the truth of $\varphi$ at $w$, denoted by $\mathcal{M}, w \models \varphi$ is defined as follows: 
    \begin{align}
    \mathcal{M}, w &\models p \iff w \in V(p) \\ 
    \mathcal{M}, w &\models \varphi \land \psi \iff \mathcal{M}, w \models \varphi \text{ and } \mathcal{M}, w \models \psi \\ 
    \mathcal{M}, w &\models \neg \varphi \iff \mathcal{M}, w \not \models \varphi \\ 
    \mathcal{M}, w &\models K_i \varphi \iff \text{for all } v \text{ such that } (w, v) \in R_i \text{ it holds } \mathcal{M}, v \models \varphi  \\ 
    \mathcal{M}, w &\models [\pi]\varphi \iff \text{for all } \alpha \in \mathcal{L}(\pi), \text{ if } \alpha \in \text{init}(Exp(w)), \text{ then } \mathcal{M}|_{\alpha}, w \models \varphi 
\end{align}
where $\alpha \in \text{init}(\pi)$ if and only if there is a $\beta \in \Sigma^*$ such that $\alpha \beta \in \mathcal{L}(\pi)$. 
\end{definition}



\section{Muddy Children Puzzle: formalization in POL}

The Muddy Children Puzzle represents a scenario where a group of children who have been playing outside is called into their house by their father. Some of the children have mud on their foreheads, others don't, and none of them know their own state. Knowing that they can all see each other, but don't know anything about their own situation, the father initially announces them that at least one child is muddy, then asks them who is certain that they have mud on their forehead. In this scenario, we assume that the children reason correctly from a logical point of view and know how to update their information based on their father's question. If a solution has not been reached yet, the father does nothing else but repeat the request until a final resolution is achieved, where all the muddy children are aware of their own state, so the ones who are not muddy will also know this fact. 

We will model the \textit{Muddy Children Puzzle} starting from the description in \cite{vanDitmarsch} (where the problem was modeled using \textit{Public announcement logic}), but using the formalism of \textit{Public observation logic}. In this case, instead of updating the model based on formulas, we will update it based on certain actions, and the main action consists of the successive questions that the father addresses.

Let $\text{\footnotesize AGENT}$ be a finite set of agents, with $|\text{\footnotesize AGENT}| = n_a$, that we will denote with the first letters of the alphabet, $a$, $b$, $c$, $\dots$ during the specification of our protocols. In order to formalize the \textit{Muddy Children Puzzle}, we will define one set of formulas $\Phi$, where $\Phi_0 \subseteq \Phi$ is the set of atomic propositions, and one set of actions $\Pi$, with $\Pi_0 \subseteq \Pi$ contains one atomic action, \textit{QF}, the father's question. Each time when the action is performed, the general knowledge is changed and the  agents enrich their individual  knowledge. We will number the father's questions to be able to update the model sequentially based on the information obtained at each step. 

For any $a \in \text{\footnotesize AGENT}$, $\Phi_0$ contains one atomic proposition $m_a$, that is true if $a$ is muddy, and false otherwise.

\begin{definition}
Given a finite set of atomic actions, $\Pi_0 = \{ \text{QF} \}$, we define the $\mathcal{L}_{\Pi}$ language by the following BNF grammar:
\begin{align}
    \pi ::= \lambda \mid
    \text{QF} \mid 
    \pi ; \pi \mid
    \pi \cup \pi \mid 
    \pi^* 
\end{align}
where $\lambda$ stands for the empty action, with a corresponding empty set of actions.
\end{definition}

We introduce the notation $\text{QF}_i := \underbrace{\text{QF} ; \dots ; \text{QF}}_{i \text{ times}}$

\begin{definition}
Given an action $\pi$, we define the corresponding set of actions $\mathcal{L}(\pi)$ as follows: 
\vspace*{0.2cm}

\begin{tabular}{ll}
    $\mathcal{L}(\lambda) = \emptyset$ & $\mathcal{L}(\pi_1 ; \pi_2) = \{ vw \mid v \in \mathcal{L}(\pi_1) \text{ and } w \in \mathcal{L}(\pi_2) \}$ \\ 
    $\mathcal{L}(\text{QF}) = \{ \text{QF} \}$ & $\mathcal{L}(\pi_1 \cup \pi_2) = \mathcal{L}(\pi_1) \cup \mathcal{L}(\pi_2)$ \\ 
    & $\mathcal{L}(\pi^*) = \{ \lambda \} \cup \bigcup_{n > 0} \mathcal{L}(\underbrace{\pi ; \pi ; \dots ; \pi}_{n \text{ times}})$
\end{tabular}
\end{definition}

In the following, let $\mathcal{M} = (W, R, V, Exp)$ be an epistemic expectations model, where the set of possible worlds $W$ consists of \textbf{ordered pairs} $(w_{a_1}, \dots, w_{n_a}) \in \{0,1\}^{n_a}$. Each $w_{a_i}$ for $a_i \in \text{\footnotesize AGENT}$ describe either if $a_i$ is muddy or not. We denote that property by the predicate $\textbf{muddy}$, defined as, for every agent $a$ and for every pair $w$,  $\textbf{muddy}(a, w) = 0$ if $a$ is not muddy and $\textbf{muddy}(a, w) = 1$ if $a$ is muddy.


\begin{definition}[State property]
    Let $a \in \text{\footnotesize AGENT}$ and let $w \in W$, where $\text{\footnotesize AGENT} = \{a_1, \dots, a, \dots, a_{n_a}\}$, where $a$ is on the $p$-th position, $0 < p \leq N$. Then $\textbf{muddy}(a, w) = \textbf{proj}_p(w)$, where \textbf{proj} is the projection function: for a pair $x = (x_1, \dots, x_n)$, $\textbf{proj}_i(x)$ returns the value of the $i$-th component $x_i$, $0 < i \leq n$.
\end{definition}

\begin{definition}[Model property]
    Let $a \in \text{\footnotesize AGENT}$. We define $\textbf{muddy}(a)$ as:
    \begin{align}
        \textbf{muddy}(a) = \left\{\begin{array}{ll}
            0, &\text{ if } \bigcap \{ \textbf{muddy}(a, w) \mid w \in W \} = \{ 0 \} \\ 
            1, &\text{ if } \bigcap \{ \textbf{muddy}(a, w) \mid w \in W \} = \{ 1 \} \\ 
            \text{undefined}, &\text{ if } \bigcap \{ \textbf{muddy}(a, w) \mid w \in W \} = \emptyset  
        \end{array}
        \right.
    \end{align}
    By abuse of notation,  when $\textbf{muddy}(a)$ is defined and there is no risk for confusion, we will use $\textbf{muddy}$ as a predicate.
    
    This predicate can also be defined as:
    \begin{align}
        \textbf{muddy}(a) = 1 &\iff \text{for all } w \in W, \mathcal{M}, w \models K_a m_a \\ 
        \textbf{muddy}(a) = 0 &\iff \text{for all } w \in W, \mathcal{M}, w \models K_a \neg m_a \\ 
        \textbf{muddy}(a) = \text{undefined} &\iff \text{otherwise
        } 
    \end{align}
\end{definition}
\begin{definition}
    Let $\mathcal{M} = (W, R, V)$ be an epistemic (Kripke) model, and $Exp : W \to \mathcal{L}_{\Pi}$ an expectation function that assigns to each state an action $\pi$ such that $\mathcal{L}(\pi) \neq \emptyset$. Then, $\mathcal{M}_{exp} = (\mathcal{M}, Exp)$ is an epistemic expectations model. 
\end{definition}



\label{def:update-session-model}
\begin{definition}
    Let $w$ be an action over $\Pi$ and let $\mathcal{M}_{Exp} = (W, R, V, Exp)$ be an expectation model. The updated model is $\mathcal{M}|_\alpha = (W', R', V', Exp')$, where $W' = \{ v \mid \mathcal{L}(Exp(v) - \alpha) \neq \emptyset \}$, $R' = R \cap W'^2$, $V' = V|_{W'}$ and for any $v \in W$, $Exp'(v) = Exp(v) - \alpha$.
\end{definition}

Each action $\text{QF}_i$, $i \geq 0$, corresponds to the $i$-th question of the father. For $i = 0$, we have $\mathcal{M}|_{\text{QF}_0} := \mathcal{M}$. Starting from $i = 1$, the states are eliminated sequentially. For $i = 1$, the state $(0,\dots,0) \in W^{n_a}$ is eliminated, as $\text{QF}_1$ is the question that announces the existence of at least one muddy child, so $\bigcup_{a \in \text{\footnotesize AGENT}} m_a$ is true: $\mathcal{M}|_{QF_1} = (W' := W - \{ 0^{n_a} \}, R|_{W'^2}, V|_{W'}, Exp)$, so after the first questions, we have states that contains at least one component that is $1$. If we continue this process, we obtain that for $\mathcal{M}|_{QF_2}$ there are removed all the states containing just one muddy agent. In general, for $\mathcal{M}|_{QF_i}$ a set of corresponding worlds with pairs $w \in W_i$ such that $w$ contain at least $i$ values of $1$: $\sum_{j} \textbf{proj}_j(w) \geq i$. This fact is supported by what happens during the father's questions: he repeats the questions because the agents knows whether the other agents are muddy or not, while none knows their own state, which means that each new question confirms to the agents that there is at least one more muddy agent they must take into account during the inference process. At this point, we have that 
\begin{align}
    \mathcal{M}|_{\text{QF}_i} = (W' := W - \{ v \in W \mid \sum_j \textbf{proj}_j(v) < i\}, R|_{W'^2}, V|_{W'}, Exp')
\end{align}

Formulas are defined by the following BNF grammar:
\begin{align}
    \varphi ::= p \mid
        \neg \varphi \mid
        \varphi \lor \varphi \mid
        K_a \varphi \mid 
        [\pi]\varphi 
\end{align}
where $p \in \Phi_0$ is an atomic formula, and $a \in \text{\footnotesize AGENT}$.

The interpretation of formulas is as follows:
\begin{align}
    \mathcal{M}, w &\models p \iff w \in V(p) \\ 
    \mathcal{M}, w &\models \varphi \land \psi \iff \mathcal{M}, w \models \varphi \text{ and } \mathcal{M}, w \models \psi \\ 
    \mathcal{M}, w &\models \neg \varphi \iff \mathcal{M}, w \not \models \varphi \\ 
    \mathcal{M}, w &\models K_i \varphi \iff \text{for all } v \text{ such that } (w, v) \in R_i \text{ it holds } \mathcal{M}, v \models \varphi  \\ 
    \mathcal{M}, w &\models [\pi]\varphi \iff \text{for all } \alpha \in \mathcal{L}(\pi), \text{ if exists } v \text{ such that } Exp(v) = \alpha ; \pi_1 ; \dots ; \pi_n, \\ \nonumber & \ \ \ \ \ \text{ then } \mathcal{M}|_\alpha, w \models \varphi 
\end{align}

The deductive system contains all instances of propositional tautologies, to which are added the following axioms (\cite{vanDitmarsch}, \cite{kozen}): 

\begin{tabular}{lll}
    $(A_1)$ & $K_a (\varphi \to \psi) \to (K_a \varphi \to K_a \psi)$ & distributivity of K \\ 
    $(A_2)$ & $K_a \varphi \to \varphi$ & truth \\
    $(A_3)$ & $K_a \varphi \to K_a K_a \varphi$ & positive introspection \\ 
    $(A_4)$ & $\neg K_a \varphi \to K_a \neg K_a \varphi$ & negative introspection \\ 
    $(A_5)$ & $[\pi](\varphi \to \psi) \to ([\pi]\varphi \to [\pi]\psi)$ & distributivity of $\pi$ \\ 
    $(A_6)$ & $[\pi](\varphi \land \psi) \leftrightarrow [\pi]\varphi \land [\pi]\psi$ & distributivity over conjunction \\ 
    $(A_7)$ & $[\pi_1; \pi_2] \varphi \leftrightarrow [\pi_1][\pi_2]\varphi$ & sequential operator \\ 
    $(A_8)$ & $[\pi_1 \cup \pi_2] \varphi \leftrightarrow [\pi_1]\varphi \land [\pi_2]\psi$ & choice \\
    $(A_9)$ & $\varphi \land [\pi][\pi^*]\varphi \leftrightarrow [\pi^*]\varphi$ & Kleene star\\
    $(A_{10})$ & $\varphi \land [\pi^*](\varphi \to [\pi]\varphi) \to [\pi^*]\varphi$ & induction axiom 
\end{tabular}

The deduction rules are: 
{ \large 
\begin{tabular}{ccc}
    $\text{\footnotesize MP }\frac{\varphi \hspace*{0.2cm} \varphi \to \psi}{\psi}$ & 
    $\text{\footnotesize NEC } \frac{\varphi}{K_a \varphi}$ & 
    $\text{\footnotesize GEN } \frac{\varphi}{[\pi]\varphi}$
\end{tabular}
}

\begin{theorem}
\label{numbermuddy}
    For $n > 1$, there are $n+1$ questions needed to discover $n$ muddy children.
\end{theorem}
\begin{proof} The proof for this theorem is well-known \cite{muddy}, but we will sketch it for our system.

    Let $\mathcal{M} = (W, R, V, Exp)$ be an epistemic model for \textit{Muddy Children Puzzle} and let $n$ be the number of muddy agents. We want to prove that, after $n + 1$ questions, every agent knows either if it is muddy or not:
    $$\mathcal{M}|_{\text{QF}_{n + 1}} \models \bigwedge_{\substack{a \in \text{\footnotesize AGENT}}} K_a m_a$$
    We will prove that by induction over $n$. 
    \begin{itemize}
        \item $\textit{base case } (n = 1) \Rightarrow$ there are 2 father's questions QF needed to discover one muddy agent. After the first one, the state $0^n$ is removed. After the second question, there is just one muddy agent $a_m$ that sees $n - 1$ non-muddy other agents, and with the $0^n$ state removed, this agent knows that is muddy:
        $$\bigwedge_{\substack{a \in \text{\footnotesize AGENT} \\ a \neq a_m}} K_{a_m} \neg m_a \land K_a m_{a_m}$$
        \item $\textit{induction step} \Rightarrow$ there are already $n$ asked questions and, based on that, $n - 1$ identified muddy agents. We want to prove that if we add a new muddy agent, there will be one additional question. In this $\mathcal{M}|_{\text{QF}_n}$ model, where is just one world left, $w \in W_{n}$, such that there are $n - 1$ values of $1$. If we have an additional agent, which we will represent as the last component of the pair, without losing generality, then we will not be able to distinguish between states $(w, 0)$ and $(w, 1)$. With a new question, the agent is able to make the inference and to figure wheter it is muddy or not.
    \end{itemize}
\end{proof}

\section{Muddy Children Puzzle with parallel sessions}

In the parallel setup, we consider that the puzzle unfolds simultaneously in multiple sessions. Each session is just a regular sequence of actions, so for each group, there is a father who addresses the questions, and the model updates based on these actions. However, there are agents who can be simultaneously present in multiple such groups, which makes them agents playing in multiple sessions simultaneously. In this case, instead of just having a number of independent parallel sessions, the sessions interact with each other through the knowledge of agents who play simultaneously in at least two sessions. If an agent finds out from one session that they are muddy (or not), they actually possess this information in all other sessions as well.

Let $\text{\footnotesize SESSION}$ be a finite set of sessions. For every session $s \in \text{\footnotesize SESSION}$, we have a finite set of agents $\text{\footnotesize AGENT}_s$. Given a finite set of atomic actions on sessions, $\Pi_0 = \{ \text{QF}^s \} |_{s \in \text{\footnotesize SESSION}}$, we define the $\mathcal{L}_{\Pi}$ language by the following BNF grammar:
\begin{align}
    \pi ::= \lambda \mid
        \text{QF}^s \mid
        \pi ; \pi \mid
        \pi \cup \pi \mid
        \pi^* 
\end{align}

where $\lambda$ stands for the empty action, $s \in \text{\footnotesize SESSION}$ and $a \in \text{\footnotesize AGENT}_s$. 

\begin{definition}[Session model]
    A session model $\mathcal{M}_s$ is defined as an expectation model $\mathcal{M}_s$, where $\mathcal{M}_s = (W_s, R_s, V_s, Exp_s)$ and $s \in \text{\footnotesize SESSION}$ is a session; $W_s$ is a finite set of worlds on the session $s$, $R_s \subseteq W_s \times W_s$ is the binary accessibility relation between worlds, $V_s : \Phi \to W_s$ is the valuation function and $Exp_s : W_s \to \mathcal{L}_\Pi$ the expectation assignment for states.
\end{definition}


\begin{definition}[Parallel model]
    A parallel model $\mathcal{M}$ is defined as $\mathcal{M} = \mathcal{M}_{s_1} \times \mathcal{M}_{s_2} \times \dots \times \mathcal{M}_{s_{n_s}} = \bigtimes_{i=1}^{n_s} \mathcal{M}_{s_i}$, where $\text{\footnotesize SESSION} = \{ s_1, s_2, \dots, s_{n_s} \}$, and for an arbitrary $s \in \text{\footnotesize SESSION}$, $\mathcal{M}_s = (W_s, R_s, V_s, Exp_s)$ is a session model.
\end{definition}

Suppose that we have a parallel model $\mathcal{M}$ that consists of $n_s$ session models, so $\mathcal{M} = \bigtimes_i \mathcal{M}_i$. We know that, for every model, we have an initial set of worlds $W_s = \{ 0, 1\}^{|\text{\footnotesize AGENT}_s|}$, where $0$ means that the property doesn't hold (for an agent $a$, we have that
$a$ is not muddy), and $1$ means that the property holds (for an agent $a$, we have that $a$ is muddy). The $(\text{\footnotesize AGENT}_s)_{s \in \text{\footnotesize SESSION}}$ sets are not mutual disjunctive: for $s_i, s_j \in \text{\footnotesize SESSION}$, $s_i \neq s_j$, is it possible that $\text{\footnotesize AGENT}_{s_i} \cap \text{\footnotesize AGENT}_{s_j} \neq \emptyset$. In that case, imagine that an action $\alpha_{s_i}$ \textbf{occurs} in the $i$-th model, $0 < i \leq |\text{\footnotesize SESSION}|$ (so we have a formula $[\beta_1; \beta_2; \dots; \alpha_{s_i}; \dots; \beta_n]\varphi$, where $\beta_1, \dots, \beta_n \in \Pi$ are also actions that occurs in the $i$-th model, and $\varphi \in \Phi$), that changes the model with respect to updated model definition; we will have that model updated, so $\mathcal{M}_i := \mathcal{M}|_{\alpha_{s_i}}$. But there is the following problem: imagine that in the new $\mathcal{M}_i$ model, we obtained a new set of worlds $W_i$ that, for an arbitrary agent $a$, we know whether $\textbf{muddy}(a)$ is either 0 or 1, and that agent also appears in other session models $\mathcal{M}_{k_1}, \mathcal{M}_{k_2}, \dots, \mathcal{M}_{k_p}$. In that case, we have to also update these models with respect to the new restrictions. 

We define $\textbf{propagate}(\alpha_{s_i}, \mathcal{M}_{s_j}) = (W'_j, R_j|_{W'^2_j}, V_j|_{W'_j}, Exp_j)$, where $W'_j$ contains all the worlds from $W_j$, but without the worlds that doesn't have the same value of the property for the agent $a$ as in $\mathcal{M}_i|_{\alpha_{s_i}}$. If the position for the agent $a$ is $i_a$ and the property value is 0 (1), we will remove from $W_j$ all the worlds that have 1 (0) on the $i_a$-th position.

\begin{definition}[Propagate actions]
Let $s_i \in \text{\footnotesize SESSION}$ and let $\alpha_{s_i}$ be an action that occurs on the $i$-th session. We define the propagation on the $\mathcal{M}_{s_j}$ models, ($s_i \neq s_j$ and for any $j$ in this case, $s_j \in \text{\footnotesize SESSION}$) by:
\begin{align}
    \textbf{propagate}(\alpha_{s_i}, \mathcal{M}_{s_j}) = (W'_j, R_j|_{W'^2_j}, V_j|_{W'_j}, Exp_j)
\end{align}
where $W'_j$ is defined as:
\begin{align}
    W'_j := W_j - \bigcup_{\substack{a \in \text{\footnotesize AGENT}_{s_i} \cap \text{\footnotesize AGENT}_{s_j} \\ \textbf{muddy}_{s_i}(a) \text{ is defined}}} \{ v \in W_j \mid \textbf{muddy}_{s_j}(a, v) \neq \textbf{muddy}_{s_i}(a)\}
\end{align}
where $\textbf{muddy}_{s_j}(a, v)$ is $\textbf{muddy}(a, v)$ on the $j$-th session, and $\textbf{muddy}_{s_i}(a)$ is $\textbf{muddy}(a)$ on the $i$-th session.
\end{definition}

\begin{definition}[Updated parallel model]
    Let $\mathcal{M}$ be a parallel model, and $\alpha_{s_i}$ an action that occurred on the $i$-th session, $0 < i \leq |\text{\footnotesize SESSION}|$. Then \begin{align}
        \mathcal{M}|_{\alpha_{s_i}} = \bigtimes_{j < i}\textbf{propagate} \big{(}\alpha_{s_i}, \mathcal{M}_{s_j}\big{)} \times \mathcal{M}|_{\alpha_{s_i}} \times \bigtimes_{i < j}\textbf{propagate} \big{(}\alpha_{s_i}, \mathcal{M}_{s_j}\big{)}
    \end{align}
    is the updated parallel model.
\end{definition}

\section{Application}

In this section, we present a concrete parallel setting, and we study how the session models interact. We have the following setup:
\begin{enumerate}
    \item $\text{\footnotesize SESSION} = \{ s_1, s_2, s_3 \}$;
    \item $\text{\footnotesize AGENT}_1 = \{ a, b \}$, $\text{\footnotesize AGENT}_2 = \{ b, c, d \}$ and $\text{\footnotesize AGENT}_3 = \{ a, d \}$;
    \item $\mathcal{M} = \mathcal{M}_1 \times \mathcal{M}_2 \times \mathcal{M}_3$;
    \item suppose that the agents $a$, $c$ and $d$ are muddy, and $b$ is not, so in $\mathcal{M}$ holds that $\textbf{muddy}(a) \land \neg \textbf{muddy}(b) \land \textbf{muddy}(c) \land \textbf{muddy}(d)$.
\end{enumerate}

In the following, we represent the models: 

{
\hspace*{2cm} 
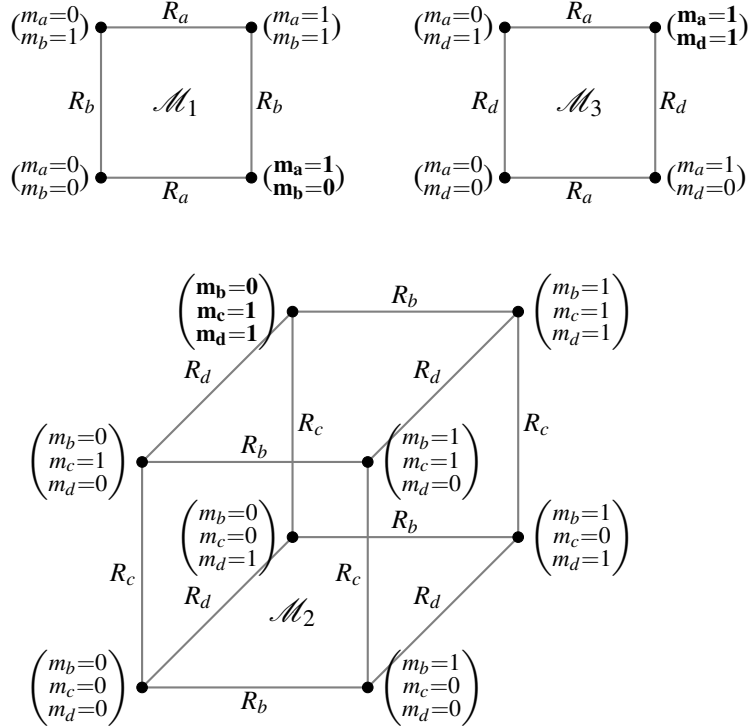
\begin{tabular}{ll}
    {
    \large 
    \begin{tikzpicture}
    \draw[gray, thick] (2, 2) -- (2, 4);
    \draw[gray, thick] (2, 4) -- (4, 4);
    \draw[gray, thick] (4, 4) -- (4, 2);
    \draw[gray, thick] (2, 2) -- (4, 2);
    \filldraw[black] (2,2) circle (2pt) node[anchor=east]{($\substack{m_a = 0 \\ m_b = 0}$)};
    \filldraw[black] (2,4) circle (2pt) node[anchor=east]{($\substack{m_a = 0 \\ m_b = 1}$)};
    \filldraw[black] (4,2) circle (2pt) node[anchor=west]{($\substack{\mathbf{m_a = 1} \\ \mathbf{m_b = 0}}$)};
    \filldraw[black] (4,4) circle (2pt) node[anchor=west]{($\substack{m_a = 1 \\ m_b = 1}$)};

    {\small \filldraw[black] (3,1.8) circle (0pt) node[anchor=center]{$R_a$};}
    {\small \filldraw[black] (3,4.2) circle (0pt) node[anchor=center]{$R_a$};}
    {\small \filldraw[black] (1.75,3) circle (0pt) node[anchor=center]{$R_b$};}
    {\small \filldraw[black] (4.25,3) circle (0pt) node[anchor=center]{$R_b$};}
    
    \filldraw[black] (3,3) circle (0pt) node[anchor=center]{$\mathcal{M}_1$};
    \end{tikzpicture}
} & {
    \large 
    \begin{tikzpicture}
    \draw[gray, thick] (2, 2) -- (2, 4);
    \draw[gray, thick] (2, 4) -- (4, 4);
    \draw[gray, thick] (4, 4) -- (4, 2);
    \draw[gray, thick] (2, 2) -- (4, 2);
    \filldraw[black] (2,2) circle (2pt) node[anchor=east]{($\substack{m_a = 0 \\ m_d = 0}$)};
    \filldraw[black] (2,4) circle (2pt) node[anchor=east]{($\substack{m_a = 0 \\ m_d = 1}$)};
    \filldraw[black] (4,2) circle (2pt) node[anchor=west]{($\substack{m_a = 1 \\ m_d = 0}$)};
    \filldraw[black] (4,4) circle (2pt) node[anchor=west]{($\substack{\mathbf{m_a = 1} \\ \mathbf{m_d = 1}}$)};

    {\small \filldraw[black] (3,1.8) circle (0pt) node[anchor=center]{$R_a$};}
    {\small \filldraw[black] (3,4.2) circle (0pt) node[anchor=center]{$R_a$};}
    {\small \filldraw[black] (1.75,3) circle (0pt) node[anchor=center]{$R_d$};}
    {\small \filldraw[black] (4.25,3) circle (0pt) node[anchor=center]{$R_d$};}
    
    \filldraw[black] (3,3) circle (0pt) node[anchor=center]{$\mathcal{M}_3$};
    \end{tikzpicture}
}
\end{tabular}
}

\vspace*{0.5cm}

{
\hspace*{2.5cm}
\begin{tikzpicture}
    \large
    \draw[gray, thick] (2, 2) -- (2, 5);
    \draw[gray, thick] (2, 5) -- (5, 5);
    \draw[gray, thick] (5, 5) -- (5, 2);
    \draw[gray, thick] (2, 2) -- (5, 2);
    \draw[gray, thick] (4, 4) -- (4, 7);
    \draw[gray, thick] (4, 7) -- (7, 7);
    \draw[gray, thick] (7, 7) -- (7, 4);
    \draw[gray, thick] (4, 4) -- (7, 4);
    \draw[gray, thick] (2, 5) -- (4, 7);
    \draw[gray, thick] (5, 5) -- (7, 7);
    \draw[gray, thick] (2, 2) -- (4, 4);
    \draw[gray, thick] (5, 2) -- (7, 4);
    \filldraw[black] (2,2) circle (2pt) node[anchor=east]{$\bigg{(}\substack{m_b = 0 \\ m_c = 0 \\ m_d = 0}\bigg{)}$};
    \filldraw[black] (2,5) circle (2pt) node[anchor=east]{$\bigg{(}\substack{m_b = 0 \\ m_c = 1 \\ m_d = 0}\bigg{)}$};
    \filldraw[black] (5,2) circle (2pt) node[anchor=west]{$\bigg{(}\substack{m_b = 1 \\ m_c = 0 \\ m_d = 0}\bigg{)}$};
    \filldraw[black] (5,5) circle (2pt) node[anchor=west]{$\bigg{(}\substack{m_b = 1 \\ m_c = 1 \\ m_d = 0}\bigg{)}$};
    \filldraw[black] (4,4) circle (2pt) node[anchor=east]{$\bigg{(}\substack{m_b = 0 \\ m_c = 0 \\ m_d = 1}\bigg{)}$};
    \filldraw[black] (4,7) circle (2pt) node[anchor=east]{$\bigg{(}\substack{\mathbf{m_b = 0} \\ \mathbf{m_c = 1} \\ \mathbf{m_d = 1}}\bigg{)}$};
    \filldraw[black] (7,4) circle (2pt) node[anchor=west]{$\bigg{(}\substack{m_b = 1 \\ m_c = 0 \\ m_d = 1}\bigg{)}$};
    \filldraw[black] (7,7) circle (2pt) node[anchor=west]{$\bigg{(}\substack{m_b = 1 \\ m_c = 1 \\ m_d = 1}\bigg{)}$};

    {\small \filldraw[black] (3.5,1.8) circle (0pt) node[anchor=center]{$R_b$};}
    {\small \filldraw[black] (3.5,5.2) circle (0pt) node[anchor=center]{$R_b$};}
    {\small \filldraw[black] (5.5,7.2) circle (0pt) node[anchor=center]{$R_b$};}
    {\small \filldraw[black] (5.5,4.2) circle (0pt) node[anchor=center]{$R_b$};}

    {\small \filldraw[black] (1.75,3.5) circle (0pt) node[anchor=center]{$R_c$};}
    {\small \filldraw[black] (4.75,3.5) circle (0pt) node[anchor=center]{$R_c$};}
    {\small \filldraw[black] (4.25,5.5) circle (0pt) node[anchor=center]{$R_c$};}
    {\small \filldraw[black] (7.25,5.5) circle (0pt) node[anchor=center]{$R_c$};}

    {\small \filldraw[black] (2.75,6.2) circle (0pt) node[anchor=center]{$R_d$};}
    {\small \filldraw[black] (5.8,6.2) circle (0pt) node[anchor=center]{$R_d$};}
    {\small \filldraw[black] (2.75,3.2) circle (0pt) node[anchor=center]{$R_d$};}
    {\small \filldraw[black] (5.8,3.2) circle (0pt) node[anchor=center]{$R_d$};}
    
    \filldraw[black] (4, 3) circle (0pt) node[anchor=center]{$\mathcal{M}_2$};
\end{tikzpicture}
}

Consider the scenario in which we run this models independently, so each agent has its own scope for every scenario. Using the results of Theorem \ref{numbermuddy}, we know that for the first model are 2 questions needed (so $\mathcal{M}_1 \models [QF_2^1] (K_a m_a \land K_b \neg m_b)$), for the second one there are 3, and for the last one also 3 questions needed. We conclude that if we analyze this scenario in a sequential manner, there are 8 question needed to solve the puzzle. Our target is to see if the number of questions can be minimised if we run these models in parallel. 

In the parallel setting, the order of actions is \textit{non-deterministic}. We will analyze a possible situation. Suppose that the first actions that occurs are the father's question for the first model. We already provided that are only two question needed to solve the puzzle, so 
\begin{align}
    \mathcal{M}_1|_{QF_2} = \big{(}\{(1, 0)\}, \emptyset, V_1|_{\{(1,0)\}}, Exp_1'\big{)}
\end{align}
Bearing in mind that we run the model in parallel, instead of just updating the $\mathcal{M}_1$ model, we have to update the entire model, $\mathcal{M}$. In that case, we have to \textbf{propagate} the new states, so $\mathcal{M}_2$ and $\mathcal{M}_3$ will be modified by keeping just the states in which $m_a$ holds and $m_b$ does not. 
\begin{align*}
    \textbf{propagate}(\text{QF}_1, \mathcal{M}_2) &= (W'_2, R_2' := R_2|_{W'_2}, V_2' := V_2|_{W'_2}, Exp_2) \\ 
    &= (W_2 - \{ v \in W_2 \mid \textbf{muddy}_{s_2}(b, v) \neq 0 \},  R_2', V_2', Exp_2) \\ 
    &= (W'_2 := \{(000), (010), (011), (001)\}, R_2', V_2', Exp_2)
\end{align*}
\begin{align*}
    \textbf{propagate}(\text{QF}_1, \mathcal{M}_3) &= (W'_3, R_3' := R_3|_{W'_3}, V_3' := V_3|_{W'_3}, Exp_3) \\ 
    &= (W_3 - \{ v \in W_3 \mid \textbf{muddy}_{s_3}(a, v) \neq 1 \}, R_3', V_3', Exp_3) \\ 
    &= ( \{ W'_3 := (10), (11) \}, R_3', V_3', Exp_3)
\end{align*}

We represent the new $\mathcal{M}|_{QF_1}$ model as:

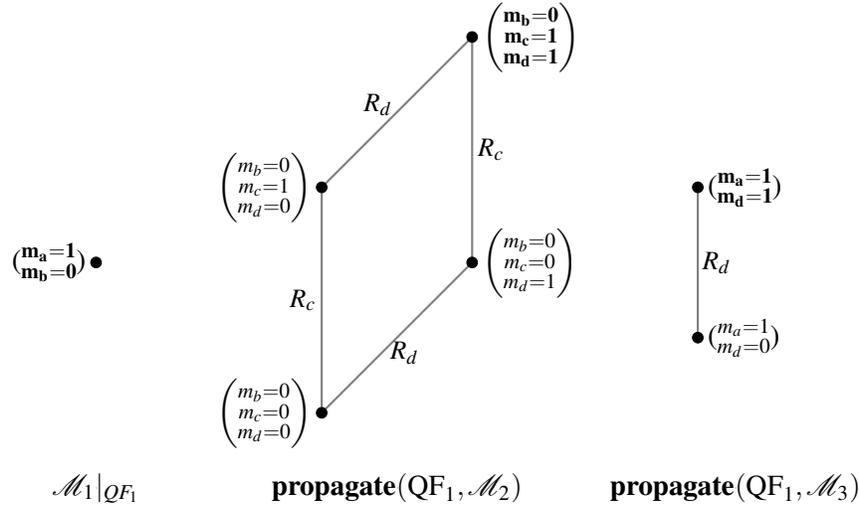
{
\hspace*{1.5cm}
\begin{tikzpicture}
    \filldraw[black] (0,4) circle (2pt) node[anchor=east]{($\substack{\mathbf{m_a = 1} \\ \mathbf{m_b = 0}}$)};
    \filldraw[black] (0,1) circle (0pt) node[anchor=center]{$\mathcal{M}_1|_{QF_1}$};
    
    \draw[gray, thick] (3, 2) -- (3, 5);
    \draw[gray, thick] (5, 4) -- (5, 7);
    \draw[gray, thick] (3, 5) -- (5, 7);
    \draw[gray, thick] (3, 2) -- (5, 4);
    \filldraw[black] (3,2) circle (2pt) node[anchor=east]{$\bigg{(}\substack{m_b = 0 \\ m_c = 0 \\ m_d = 0}\bigg{)}$};
    \filldraw[black] (3,5) circle (2pt) node[anchor=east]{$\bigg{(}\substack{m_b = 0 \\ m_c = 1 \\ m_d = 0}\bigg{)}$};
    \filldraw[black] (5,4) circle (2pt) node[anchor=west]{$\bigg{(}\substack{m_b = 0 \\ m_c = 0 \\ m_d = 1}\bigg{)}$};
    \filldraw[black] (5,7) circle (2pt) node[anchor=west]{$\bigg{(}\substack{\mathbf{m_b = 0} \\ \mathbf{m_c = 1} \\ \mathbf{m_d = 1}}\bigg{)}$};

    {\small \filldraw[black] (2.75,3.5) circle (0pt) node[anchor=center]{$R_c$};}
    {\small \filldraw[black] (5.25,5.5) circle (0pt) node[anchor=center]{$R_c$};}
    {\small \filldraw[black] (3.75,6.1) circle (0pt) node[anchor=center]{$R_d$};}
    {\small \filldraw[black] (4.1,2.8) circle (0pt) node[anchor=center]{$R_d$};}
    
    \filldraw[black] (4, 1) circle (0pt) node[anchor=center]{$\textbf{propagate}(\text{QF}_1, \mathcal{M}_2)$};

    \draw[gray, thick] (8, 5) -- (8, 3);
    \filldraw[black] (8,3) circle (2pt) node[anchor=west]{($\substack{m_a = 1 \\ m_d = 0}$)};
    \filldraw[black] (8,5) circle (2pt) node[anchor=west]{($\substack{\mathbf{m_a = 1} \\ \mathbf{m_d = 1}}$)};

    {\small \filldraw[black] (8.25,4) circle (0pt) node[anchor=center]{$R_d$};}
    
    \filldraw[black] (8.5,1) circle (0pt) node[anchor=center]{$\textbf{propagate}(\text{QF}_1, \mathcal{M}_3)$};
\end{tikzpicture}
}

Suppose that the next action is father's question in $\mathcal{M}_3$. The agent $a$ already knows that $\textbf{muddy}(a)$ holds, so by seeing $d$ can answer the question by $\textbf{muddy}(a) \land \textbf{muddy}(d)$. We also have a \textbf{propagation}, that transforms $\mathcal{M}_2$ into a model with two worlds, $W_2 = \{ (0, 1, 1), (0, 0, 1) \}$. After the father's question, both $b$ and $d$ cand solve the puzzle, by seeing that $c$ is also muddy. In this setting, we only needed four father's questions instead of eight. 

\section*{Conclusion and Further Work} 

The formalism for representing parallel models presented in this article, using the \textit{Muddy Children Puzzle} as an example, is based on POL (public observation logic \cite{hidden}) and builds upon the modeling in PAL (public announcement logic \cite{vanDitmarsch}). In POL, observations are seen as strings over an alphabet $\Sigma$, which we replaced in our formalism with actions. For the specific problem we presented, the only actions considered are represented by the father's questions $\text{QF}$. Although our aim was to demonstrate how information can be propagated in parallel sessions for a simple case, this approach can be generalized and applied to the analysis of security protocols.

In \cite{delp}, we introduced a formalism for analyzing security protocols that use symmetric key cryptography (called DELP, \textit{dynamic epistemic logic for protocols}), also starting from POL but providing a different semantics for actions ($[\pi]\varphi$). An idea for future work is to combine DELP with the method of transmitting information in parallel across multiple sessions to study various types of attacks. This integration could enable a more comprehensive analysis of security protocols, incorporating the dynamics of information updates and potential attacks in a parallel setting. By combining the concepts from DELP and the parallel model, security analysts can investigate how various security protocols perform under different conditions and explore possible vulnerabilities in a more intricate scenario.

\bibliography{generic}
\bibliographystyle{eptcs}

\end{document}